\documentclass{article}
\usepackage{ijcai19}

\usepackage{times}
\usepackage{soul}
\usepackage{url}
\usepackage[hidelinks]{hyperref}
\usepackage[utf8]{inputenc}
\usepackage[small]{caption}
\usepackage{graphicx}
\usepackage{amsmath}
\usepackage{booktabs}
\usepackage{subfigure}
\usepackage{diagbox}
\usepackage{mathrsfs}
\usepackage{weiwAlgorithm}
\usepackage{amssymb}
\usepackage{amsthm}

\newtheorem{lemma}{Lemma}

\newtheorem{theorem}{Theorem}

\newtheorem{definition}{Definition}

\newcommand{\ie}{\emph{i.e.,}\xspace}

\newcommand{\etal}{et al.\xspace}

\newcommand{\edge}[1]{(#1)}
\newcommand{\set}[1]{\{#1\}}
\newcommand{\Truss}[1]{T_{#1}\xspace}
\newcommand{\Node}[1]{V_{#1}\xspace}
\newcommand{\Edge}[1]{E_{#1}\xspace}
\newcommand{\Nei}[2]{N(#1,#2)\xspace}
\newcommand{\Deg}[2]{deg(#1,#2)\xspace}
\newcommand{\Core}[1]{C_{#1}\xspace}
\newcommand{\Sup}[2]{sup(#1,#2)\xspace}
\newcommand{\Tdel}[2]{T_{#1}^{#2}\xspace}

\DeclareMathOperator*{\argmax}{arg\,max}

\newcommand\blfootnote[1]{%
  \begingroup
  \renewcommand\thefootnote{}\footnote{#1}%
  \addtocounter{footnote}{-1}%
  \endgroup
}

\title{Critical Edge Identification: A K-Truss Based Model}

\author{
Weijie Zhu$^{1,3}$\and
Mengqi Zhang$^2$\and
Chen Chen$^{1*}$\and
Xiaoyang Wang$^2$\and
Fan Zhang$^4$\and
Xuemin Lin$^{1,3,4}$
\affiliations
$^1$East China Normal University, China\\
$^2$Zhejiang Gongshang University, China\\
$^3$Zhejiang Lab, Hangzhou, China\\
$^4$The University of New South Wales, Australia
\emails
\{weijie.zhu93, mengqiz.zjgsu, fanzhang.cs\}@gmail.com,
\{chenc, xiaoyangw\}@zjgsu.edu.cn,
lxue@cse.unsw.edu.au
}

\begin{document}

\maketitle

\begin{abstract}
In a social network, the strength of relationships between users can significantly affect the stability of the network. In this paper, we use the $k$-truss model to measure the stability of a social network. To identify critical connections, we propose a novel problem, named $k$-truss minimization. Given a social network $G$ and a budget $b$, it aims to find $b$ edges for deletion which can lead to the maximum number of edge breaks in the $k$-truss of $G$. We show that the problem is NP-hard. To accelerate the computation, novel pruning rules are developed to reduce the candidate size. In addition, we propose an upper bound based strategy to further reduce the searching space. Comprehensive experiments are conducted over real  social networks to demonstrate the efficiency and effectiveness of the proposed techniques.\blfootnote{*Corresponding author}
\end{abstract}

\section{Introduction}
\label{sec:intro}

As a key problem in graph theory and social network analysis, the mining of cohesive subgraphs, such as $k$-core, $k$-truss, clique, etc, has found many important applications in real life~\cite{cohen2008trusses,tsourakakis2013denser,DBLP:conf/icde/WenQZLY16,yu2013more}.
The mined cohesive subgraph can serve as an important metric to evaluate the properties of a network, such as network engagement.
In this paper, we use the $k$-truss model to measure the cohesiveness of a social network.
Unlike $k$-core, $k$-truss not only emphasizes the users' engaged activities (\ie number of friends), but also requires strong connections among users. That is, the $k$-truss of $G$ is the maximal subgraph where each edge is involved in at least $k-2$ triangles.
Note that triangle is an important building block for the analysis of social network structure~\cite{xiao2017asymptotic,cui2018efficient}.
Thus the number of edges in the $k$-truss can be utilized to measure the stability of network structure.

\begin{figure}[htb]
\centering
\includegraphics[width=0.6\linewidth]{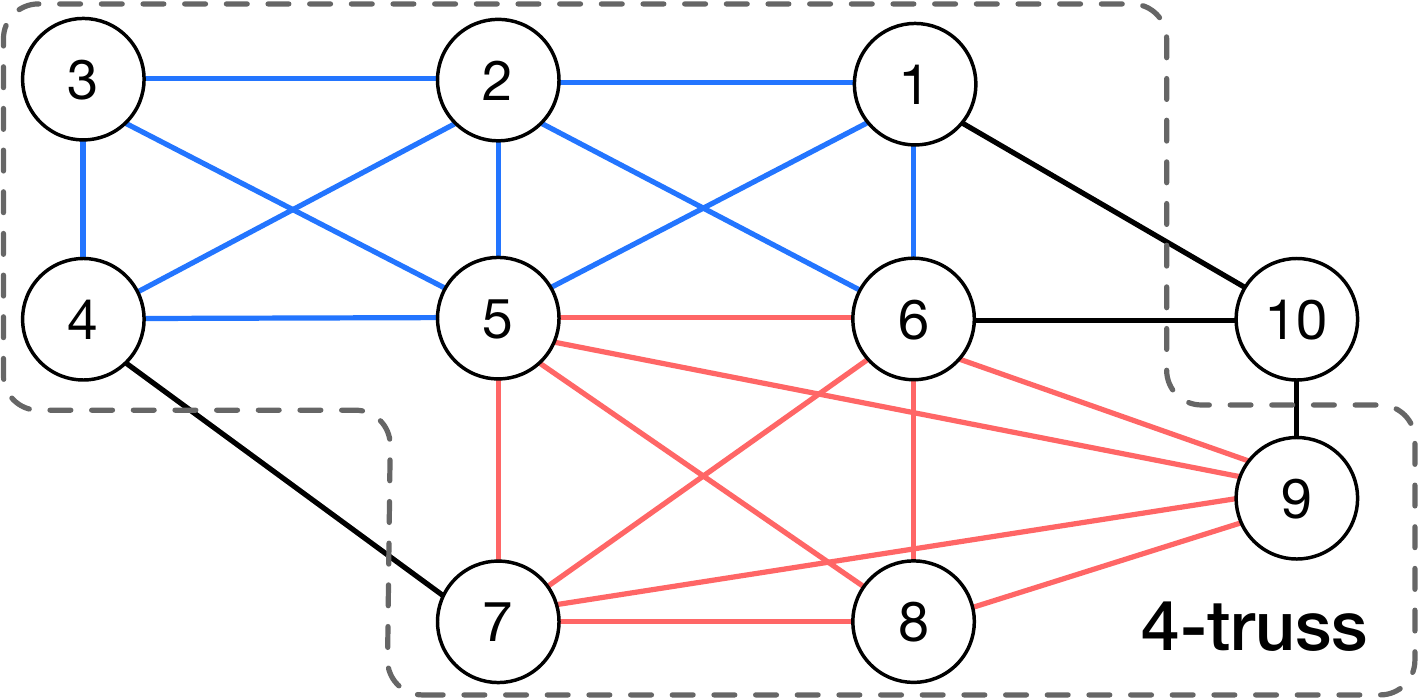}
\caption{Motivation Example}
\label{fig:moti}
\end{figure}

The breakdown of a strong connection may affect other relationships, which can make certain relationships involved in less than $k-2$ triangles and removed from the $k$-truss. Hence, it will lead to a cascading breakdown of relationships eventually. To identify the critical edges, in this paper, we investigate the $k$-truss minimization problem. Given a social network $G$ and a budget $b$, $k$-truss minimization aims to find a set $B$ of $b$ edges, which will result in the largest number of edge breaks in the $k$-truss by deleting $B$.

Figure~\ref{fig:moti} is a toy social network with 10 users. Suppose $k$ is 4. Then only the blue and red edges belong to the $4$-truss. If we delete edge $\edge{v_2, v_5}$, it will affect the connections
among other users and lead to the removal of all the blue edges from the $4$-truss, since they no longer meet the requirement of $4$-truss. We can see that the deletion of one single edge can seriously collapse the social network.
The $k$-truss minimization problem can find many applications in real life. For instance,
given a social network, we can reinforce the community by paying more attention to the critical relationships. Also, we can strengthen the important connections to enhance the stability
of a communication network or detect vital connections in enemy's network for military purpose.

The main challenges of this problem lie in the following two aspects. Firstly, we prove that the problem is NP-hard. It means that it is non-trivial to obtain the result in polynomial time. Secondly, the number of edges in a social network is usually quite large. Even if we only need to consider the edges in $k$-truss as candidates, it is still a large amount of edges to explore. To the best of our knowledge, we are the first to investigate the $k$-truss minimization problem through edge deletion.
We formally define the problem and prove its hardness. Novel pruning rules are developed to reduce the searching space. To further speed up the computation, an upper bound based strategy is proposed. 

\section{Preliminaries}
\label{sec:preli}

\subsection{Problem Definition}
\label{sec:definition}

We consider a social network $G$ as an undirected graph. Given a subgraph $S \subseteq G$, we use
$\Node{S}$ (resp. $\Edge{S}$) to denote the set of nodes (resp. edges) in $S$. $\Nei{u}{S}$ is
the neighbors of $u$ in $S$. $\Deg{u}{S}$ equals $|\Nei{u}{S}|$, denoting the degree of $u$ in $S$. $m = |\Edge{G}|$ is the number of edges in $G$.
Assuming the length of each edge equals 1, a triangle is a cycle of length 3 in the graph.
For $e \in \Edge{G}$, a \textbf{containing-\emph{e}-triangle} is a triangle which contains $e$.

\begin{definition}[\textbf{$k$-core}]
Given a graph $G$, a subgraph $S$ is the k-core of $G$, denoted as $\Core{k}$, if (i) $S$ satisfies degree constraint, i.e., $\Deg{u}{S} \geq k$ for every $u \in \Node{S}$; and (ii) $S$ is maximal, i.e., any supergraph of $S$ cannot be a k-core.
\end{definition}

\begin{definition}[\textbf{edge support}]
Given a subgraph $S \subseteq G$ and an edge $e \in \Edge{S}$, the edge support of $e$ is the number of containing-e-triangles in $S$, denoted as $\Sup{e}{S}$.
\end{definition}

\begin{definition}[\textbf{$k$-truss}]
Given a graph $G$, a subgraph $S$ is the $k$-truss of $G$, denoted by $\Truss{k}$, if (i) $\Sup{e}{S} \geq k-2$ for every edge $e \in \Edge{S}$; (ii) $S$ is maximal, i.e., any supergraph of $S$ cannot be a $k$-truss; and (iii) $S$ is non-trivial, i.e., no isolated node in $S$.
\end{definition}

\begin{definition}[\textbf{trussness}]
The trussness of an edge $e \in \Edge{G}$, denoted as $\tau(e)$, is the largest integer $k$  that satisfies $e \in \Edge{\Truss{k}}$ and $e \notin \Edge{\Truss{k+1}}$.
\end{definition}

Based on the definitions of $k$-core and $k$-truss, we can see that $k$-truss not only requires sufficient number of neighbors, but also has strict constraint over the strength of edges.
A $k$-truss is at least a ($k$-1)-core. Therefore, to compute the $k$-truss, we can first compute the ($k$-1)-core and then find the $k$-truss over ($k$-1)-core by iteratively removing all the edges that violate the $k$-truss constraint. The time complexity is $O(m^{1.5})$~\cite{Wang:2012:TDM:2311906.2311909}.
Given a set $B$ of edges in $G$, we use $\Tdel{k}{B}$ to denote the $k$-truss after deleting $B$. We use $|\Tdel{k}{B}|$ to denote the number of edges in $\Tdel{k}{B}$. We define the followers $F(B,\Truss{k})$ of $B$ as the edges that are removed from $\Truss{k}$ due to the deletion of $B$. Then our problem can be formally defined as follows.

Given a graph $G$ and a budget $b$, the \textbf{$k$-truss minimization problem} aims to find a set $B^*$ of $b$ edges, such that the $|\Tdel{k}{B^*}|$ is minimized. It is also equivalent to finding an edge set $B^*$ that can maximize $|F(B^*,\Truss{k})|$, \ie
\begin{equation*}
B^* = \argmax_{B \subseteq \Edge{G} \wedge  |B| = b}|F(B,\Truss{k})|.
\end{equation*}

\begin{figure}[t]
\begin{center}
 \subfigure[\small{Constructed Example for NP-hard Proof}]{
    \includegraphics[width=1\columnwidth]{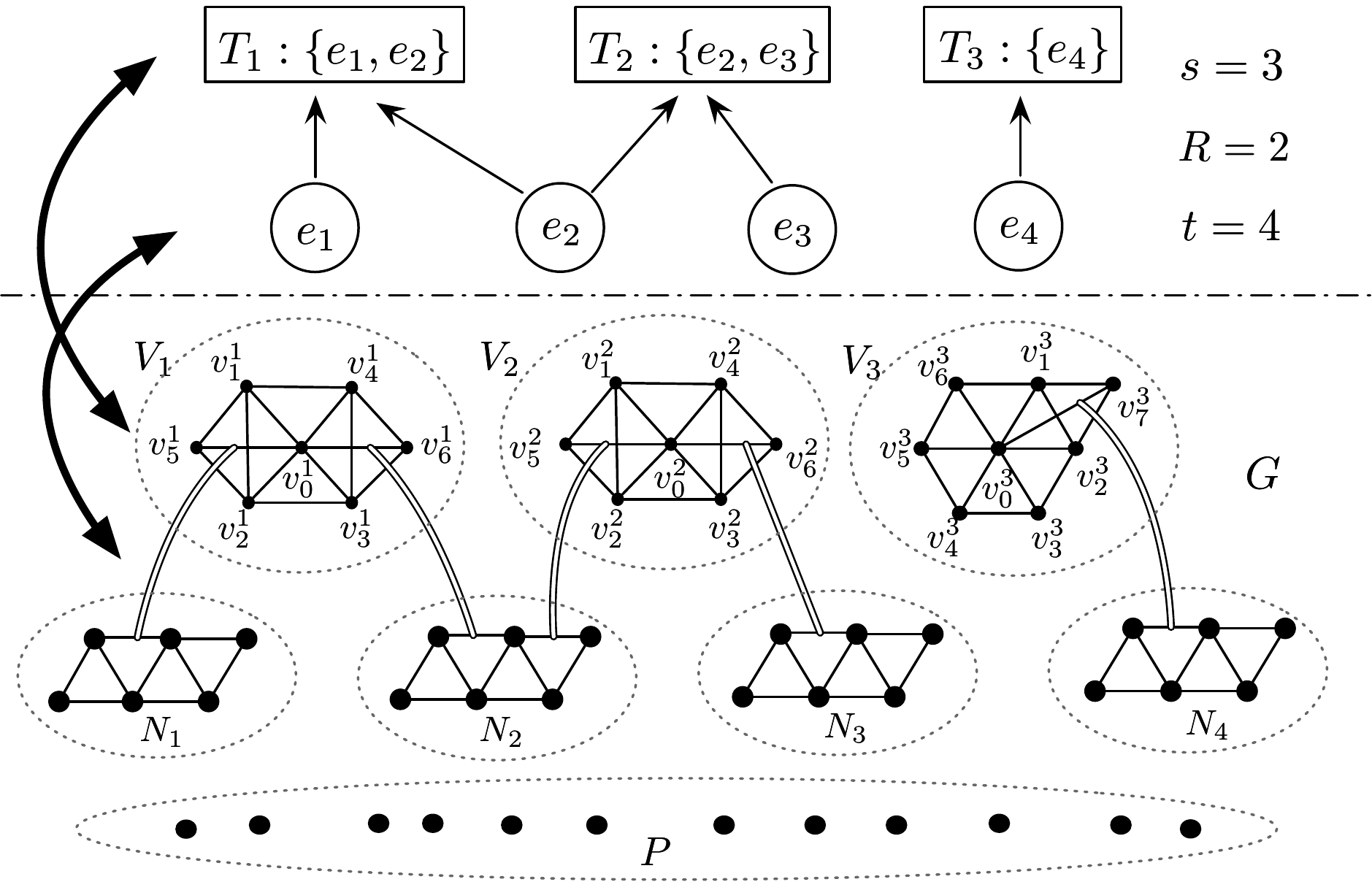}
    \label{fig:np}
  }

  \subfigure[\small{Structure Illustration}]{
    \includegraphics[width=0.46\columnwidth]{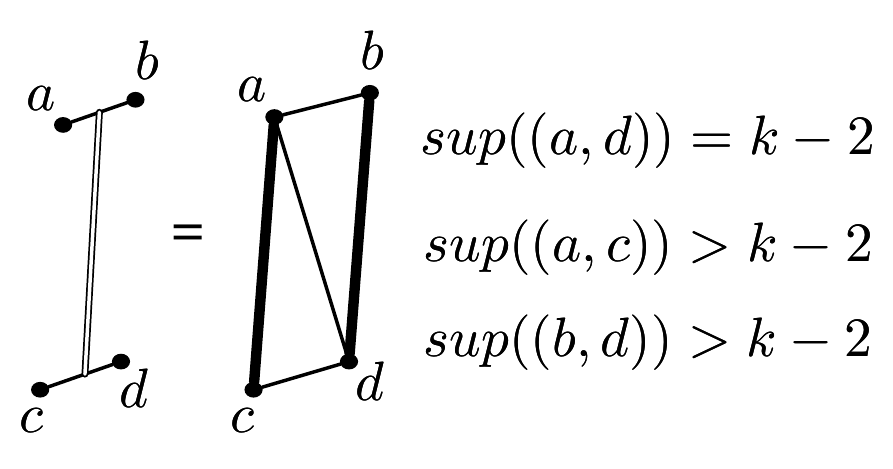}
    \label{fig:note}
  }
  \subfigure[\small{Construction of $V$ for $R = 3$}]{
    \includegraphics[width=0.46\columnwidth]{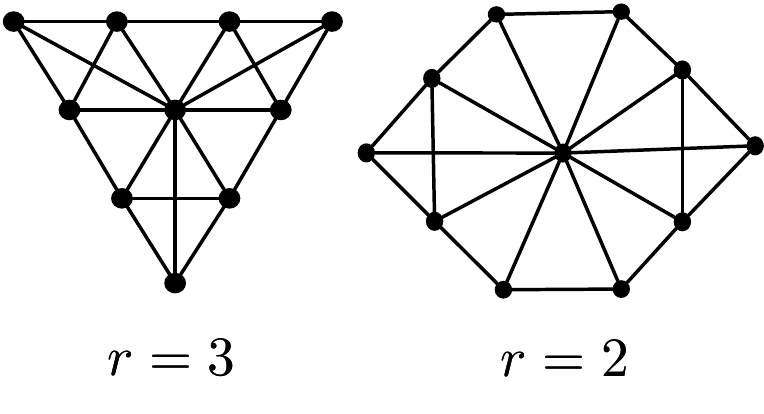}
    \label{fig:construction_example}
  }
\end{center}
\caption{Example for NP-hard}
\label{fig:np_np}
\end{figure}

According to Theorem~\ref{th:nphard} and~\ref{th:nonsub}, the $k$-truss minimization problem is NP-hard for $k \geq 5$, and the objective function is monotonic but not submodular.

\begin{theorem}
\label{th:nphard}
For $k \geq 5$, the $k$-truss minimization problem is NP-hard.
\end{theorem}

\begin{proof}
For $k \geq 5$, we sketch the proof for $k=5$. A similar construction can be applied for the case of $k >5$.
When $k = 5$, we reduce the $k$-truss minimization problem from the maximum coverage problem~\cite{DBLP:conf/coco/Karp72}, which aims to find $b$ sets to cover the largest number of elements, where $b$ is a given budget. We consider an instance of maximum coverage problem with $s$ sets $T_{1}$, $T_{2}$, .., $T_{s}$ and $t$ elements $\{e_{1}, .., e_{t}\}$ = $\cup_{1 \leq i \leq s}T_{i}$. We assume that the maximum number of elements inside $T$ is $R \leq t$. Then we construct a corresponding instance of the $k$-truss minimization problem in a graph $G$ as follows. Figure~\ref{fig:np} is a constructed example for $s=3,t=4,R=2$.

We divide $G$ into three parts, $V$, $N$ and $P$. 1) $V$ consists of $s$ parts. Each part $V_i$ corresponds to $T_{i}$ in the maximum coverage problem instance. 2) $N$ consists of $t$ parts. Each part $N_i$ corresponds to $e_i$ in the maximum coverage problem instance. 3) $P$ is a dense subgraph. The support of edges in $P$ is no less than $k - 2 + b$. Specifically, suppose $T_i$ consists of $r_i \leq R$ elements, $V_i$ consists of $4R - r_i +1$ nodes and $8R - r_i$ edges. To construct $V_i$, we first construct a $(4R-2r_i)$-polygon. Then, we add a node $v^i_0$ in the center of $(4R-2r_i)$-polygon and add $4R-2r_i$ edges between $v^i_0$ and $v^i_1, ..., v^i_{4R-2r_i}$. Finally, we further add $r_i$ nodes $v^i_{4R-2r_i+1}, ..., v^i_{4R-r_i}$ and $3r_i$ edges $\{(v^i_{0}, v^i_{4R-2r_i+1}), (v^i_{1}, v^i_{4R-2r_i+1}), (v^i_{2}, v^i_{4R-2r_i+1}),..., (v^i_{0},$ $v^i_{4R-r_i}), (v^i_{2r_i - 1}, v^i_{4R-r_i}), (v^i_{2r_i}, v^i_{4R-r_i})\}$. With the construction, the edges in $V$ have support no larger than 3. We use $P$ to provide support for edges in $V$ and make the support of edges in $V$ to be 3. Each part in $N$ consists of $2R+2$ nodes and the structure is a list of $4R$ triangles which is shown in Figure~\ref{fig:np}. For each element $e_i$ in $T_j$, we add two triangles between $N_i$ and $V_j$ to make them triangle connected. The structure is shown in ~\ref{fig:note}. Note that each edge in $N_i$ and $V_j$ can be used at most once. We can see that edges in $N$ have support no larger than 3. Finally, we use $P$ to provide support for edges in $N$ and make the support of edges in $N$ to be 3. Then the construction is completed.
The construction of $V_i$ for $R = 3$ is shown in Figure~\ref{fig:construction_example}.

With the construction, we can guarantee that 1) deleting any edge in $V_i$ can make all the edges in $V_i$ and the edges in $N_j$ who have connections with $V_i$ deleted from the truss. 2) Only the edges in $V_i$ can be considered as candidates. 3) Except the followers in $N$, each $V_i$ has the same number of followers. In Figure~\ref{fig:np}, deletion of each $V_i$ can make $8R$ edges (except the edges in $N$) removed. Consequently, the optimal solution of $k$-truss minimization problem is the same as the maximum coverage problem. Since the maximum coverage problem is NP-hard, the theorem holds.
\end{proof}

\begin{theorem}
\label{th:nonsub}
The objective function $f(x) = |F(x,\Truss{k})|$ is monotonic but not submodular.
\end{theorem}

\begin{proof}
Suppose $B \subseteq B'$. For every edge $e$ in $F(B,\Truss{k})$, $e$ will be deleted from the $k$-truss when deleting $B'$. Thus $f(B) \leq f(B')$ and $f$ is monotonic.
Given two sets $A$ and $B$, if $f$ is submodular, it must hold that $f(A \cup B) + f(A \cap B) \leq f(A) + f(B)$. We show that the inequality does not hold by constructing a counter example. In Figure~\ref{fig:moti},
for $k=4$, suppose $A=\{(v_{5},v_{6})\}$ and $B=\{(v_{7},v_{8})\}$. We have $f(A)=4$, $f(B)=0$, $f(A \cup B)=12$ and $f(A \cap B)=0$.
The inequation does not hold. $f$ is not submodular.
\end{proof}

\begin{algorithm}[t]
{
\Input{$G$: a social network, $k$: truss constraint, $b$: the budget}
\Output{$B$: the set of deleted edges}
\SetArgSty{text}
\SetFuncSty{text}
\caption{Baseline Algorithm}
\label{alg:greedy}
\State{$B \leftarrow \emptyset$ ; $\Truss{k} \leftarrow k$-truss of $G$}
\While{$|B| < b$}
{

    \State{$e^* \gets \argmax_{e \in \Edge{\Truss{k}}} |F(e, \Truss{k})|$ }
    \State{delete $e^*$ from $\Truss{k}$ and update $\Truss{k}$ }
    \State{$B \gets B \cup \set{e^\ast}$ }
}
\Return{$B$ }
}
\end{algorithm}

\subsection{Baseline Algorithm}
\label{sec:gen_framework}

For the \emph{k}-truss minimization problem, a naive solution is to enumerate all the possible edge sets of size $b$, and return the best one. However, the size of a real-world social network is usually very large. The number of combinations is enormous to enumerate. Due to the complexity and non-submodular property of the problem, we resort to the greedy framework.
Algorithm~\ref{alg:greedy} shows the baseline greedy algorithm.
It is easy to verify that we only need to consider the edges in the $k$-truss as candidates.
The algorithm iteratively finds the edge with the largest number of followers in the current $k$-truss (Line 3). The algorithm terminates when $b$ edges are found. The time complexity of the baseline algorithm is $O(bm^{2.5})$.

\begin{algorithm}[!th]
\SetAlgoVlined
\Input{$G$: a social network, $k$: truss constraint, $b$: the budget}
\Output{$B$: the set of deleted edges}
\SetArgSty{text}
\SetFuncSty{text}
\caption{Group based Algorithm}
\label{alg:improved}
\StateCmt{$B \gets \emptyset$ ; $\Truss{k} \gets$ $k$-truss of $G$ }{compute $k$-truss}
\While{$|B| < b$}
{
    \State{mark all edges in $\Truss{k}$ as \textit{unvisited} }
    \StateCmt{$T \gets$ FindGroup ($\Truss{k}$) }{Line 12-19}
    \For{each $e$ in $T$}
    {
        \State{compute $F(e, \Truss{k})$ }
        \State{$T \leftarrow T \backslash F(e, \Truss{k})$ }
    }
    \State{$e^\ast \gets$ the edge with the most followers }
    \State{update $k$-truss $\Truss{k}$ }
    \State{$B\leftarrow B \cup \set{e^\ast}$ }
}
\Return{$B$ }

\SetKwFunction{FMain}{FindGroup }
    \SetKwProg{Fn}{Function}{:}{}
    \Fn{\FMain{$S$}}{
        \StateCmt{$C \leftarrow \emptyset$ ; gID  $\leftarrow 0$}{$C$ stores the candidates}
        \For{each $e \in S$ }
        {\If{$sup(e, S)=k-2$ and $e$ is \textit{unvisited}}
            {
                \StateCmt{GroupExpansion ($S, e$)}{Line 20-32}
                \State{gID$++$ }
            }

        }
        \textbf{return} $C$
     }
\textbf{End Function}

\SetKwFunction{FMain}{GroupExpansion }
    \SetKwProg{Fn}{Function}{:}{}
    \Fn{\FMain{$S, e$}}{
        \State{$Q \gets \emptyset$ $; Q.enqueue(e)$; mark $e$ as \textit{visited} }
        \While{$Q \neq \emptyset$}
        {
            \State{$e'(u, v) \leftarrow Q.dequeue()$}
            \ForEach{$a \in N(u,\Truss{k}) \cap N(v,\Truss{k})$}
            {
                    \If{$\edge{u, a}$ is \textit{unvisited} and $sup= k-2$ }
                    {
                        \State{$Q.enqueue((u, a))$ }
                        \State{mark $\edge{u,a}$ as \textit{visited} }
                    }
                    \If{$(v, a)$ is \textit{unvisited} and $sup= k-2$ }
                    {
                        \State{$Q.enqueue((v, a))$ }
                        \State{mark $\edge{v,a}$ as \textit{visited} }
                    }
                    \State{update $C$ }

                }
        }
     }
\textbf{End Function}
\end{algorithm}

\section{Group Based Solution}
\label{sec:framework}

In this section, novel pruning techniques are developed to accelerate the search in baseline algorithm.

\subsection{Candidate Reduction}

Before introducing the pruning rules, we first present some definitions involved.

\begin{definition}[\textbf{triangle adjacency}]
Given two triangles $\triangle_1,$ $\triangle_2$ in $G$, they are triangle adjacent if $\triangle_1$ and $\triangle_2$ share a common edge, which means $\triangle_1 \cap \triangle_2 \neq \varnothing$.
\end{definition}

\begin{definition} [\textbf{triangle connectivity}]
Given two triangles $\triangle_s, \triangle_t$ in $G$, they are triangle connected, denoted as $\triangle_s \leftrightarrow \triangle_t$, if there exists a sequence of $\theta$ triangles $\triangle_1, \triangle_2, ..., \triangle_{\theta}$ in $G$, such that $\triangle_s = \triangle_1, \triangle_t = \triangle_{\theta}$, and for $1 \leq i < \theta$, $\triangle_i$ and $\triangle_{i+1}$ are triangle adjacent.
\end{definition}

For two edges $e$ and $e'$, we say they are \textbf{triangle adjacent},
if $e$ and $e'$ belong to the same triangle.
As shown in the baseline algorithm, we only need to consider the edges in $\Truss{k}$ as candidates.
Lemma~\ref{le:rule1} shows that we only need to explore the edges in $Q$.

\begin{lemma}\label{le:rule1}
Given a $k$-truss $\Truss{k}$, let $P = \{e~|~\Sup{e}{\Truss{k}}$ $= k - 2\}$. If an edge $e$ has at least one follower, $e$ must be in $Q$, where $Q = \{e~|~e \in \Truss{k} \wedge \exists e' \in P$ where $e$ and $e'$ are triangle adjacent$\}$.
\end{lemma}

\begin{proof}
We prove the lemma by showing that edges in $\Edge{G} \setminus Q$ do not have followers. We divide $\Edge{G} \setminus Q$ into two sets. 1) For edge with trussness less than $k$, it will be deleted during the $k$-truss computation. 2) For an edge $e$ in $\Truss{k}$, if $e$ is not triangle adjacent with
any edge in $P$, it means $e$ is triangle adjacent with edges such as $e'$ whose $\Sup{e'}{\Truss{k}} > k-2$. If we delete $e$, all the edges triangle adjacent with $e$ will still have support
 at least $k-2$ in $\Truss{k}$. Thus, $e$ has no follower. The lemma is correct.
\end{proof}

Based on Lemma~\ref{le:rule2}, we can skip the edges that are the followers of the explored ones.

\begin{lemma}\label{le:rule2}
Given two edges $e_1, e_2 \in \Truss{k}$, if $e_1 \in F(e_2, \Truss{k})$, then we have $F(e_1, \Truss{k}) \subseteq F(e_2, \Truss{k})$.
\end{lemma}

\begin{proof}
$e_1 \in F(e_2, \Truss{k})$, it implies that $e_1$ will be deleted during the deletion of $e_2$. Therefore, each edge in $F(e_1, \Truss{k})$ will be deleted when $e_2$ is deleted. Consequently, we have $F(e_1, \Truss{k}) \subseteq F(e_2, T_k)$.
\end{proof}

To further reduce the searching space, we introduce a pruning rule based on $k$-support group.

\begin{definition}[\textbf{$k$-support group}]
Given a k-truss $\Truss{k}$, a subgraph $S \subseteq \Truss{k}$ is a $k$-support group if it satisfies : 1) $\forall e$ $\in S$, $\Sup{e}{\Truss{k}} = k-2$. 2)
$\forall e_1, e_2 \in S$, suppose $e_1 \in \triangle_s$, $e_2 \in \triangle_t$. There exists a sequence of $\theta \geq 2$ triangles $\triangle_1, ..., \triangle_{\theta}$ with $\triangle_s = \triangle_1$, $\triangle_t = \triangle_{\theta}$. For $i \in [1,\theta)$, $\triangle_i $ $\cap \triangle_{i+1} = {e}$ and $\Sup{e}{\Truss{k}} = k-2$. 3) $S$ is maximal, i.e., any supergraph of $S$ cannot be a $k$-support group.
\end{definition}

Lemma~\ref{le:rule3} shows that edges in the same $k$-support group are equivalent. The deletion of any edge in a $k$-support group can lead to the deletion of the whole $k$-support group.

\begin{lemma}\label{le:rule3}
$S$ is a $k$-support group of $\Truss{k}$. For $\forall e \in S$, if we delete $e$, we can have $S$ deleted from $\Truss{k}$.
\end{lemma}

\begin{proof}
Since $S$ is a $k$-support group of $\Truss{k}$, for $\forall e, e' \in S$, suppose that $e \in \triangle_s, e' \in \triangle_t$, there exists a sequence of $\theta$ triangles $\triangle_1, ..., \triangle_{\theta}$ with $\triangle_s = \triangle_1, \triangle_t = \triangle_{\theta}$. For $i \in [1,\theta)$, $\triangle_i \cap \triangle_{i+1} = {e_i}$ and $\Sup{e_i}{\Truss{k}} = k-2$. The deletion of any edge inside the group will destroy the corresponding triangles and decrease the support of triangle adjacent edges by 1.
It will lead to a cascading deletion of subsequent triangle edges in the group due to the violation of truss constraint. Therefore, the lemma holds.
\end{proof}

According to Lemma~\ref{le:rule3}, we only need to add one edge from a $k$-support group to the candidate set, and the other edges in the group can be treated as the followers of the selected edge.
In the following lemma, we can further prune the edges that are adjacent with multiple edges in a $k$-support group.

\begin{lemma}\label{le:rule4}
Suppose that $e \in \Truss{k}$ and $\Sup{e}{\Truss{k}} = w > k-2$.
For a $k$-support group $S$, if $e$ belongs to more than $w - k + 2$ triangles, each of which contains at least one edge in $S$, then $e$ is a follower of $S$.
\end{lemma}

\begin{proof}
According to Lemma~\ref{le:rule3}, by removing an edge from $S$, we have $S$ deleted from $\Truss{k}$. Since $e$ belongs to more than $w - k + 2$ triangles, each of which contains at least one edge in $S$, the support of $e$ will decrease by more than $w-k+2$ due to the deletion of $S$. So its support will be less than $k-2$ and it will be deleted due to the support constraint. Thus, $e$ is a follower of $S$.
\end{proof}

\subsection{Group Based Algorithm}

We improve the baseline algorithm by integrating all the pruning rules above,
and the details are shown in Algorithm~\ref{alg:improved}.
In each iteration, we first find $k$-support groups of current $\Truss{k}$ and compute the candidate set $T$ according to Lemma~\ref{le:rule3} (Line 4).
This process, \ie \textbf{FindGroup} function, corresponds to Line 12-19. It can be done by conducting
BFS search from edges in $\Truss{k}$. We use a hash table to maintain the group id (\ie gID) for each edge and the gID starts from 0 (Line 13). For each unvisited edge with support of $k-2$, we conduct
a BFS search from it by calling function \textbf{GroupExpansion} (Line 20-32).
During the BFS search, we visit the edges that are triangle adjacent with the current edge, and push the edges with support of $k-2$ into the queue if they are not visited (Line 25 and 28).
The edges, which are visited in the same BFS round, are marked with the current gID.
For the visited edges with support larger than $k-2$, we use a hash table to record its coverage with the current $k$-support group, and update the candidate set based on Lemma~\ref{le:rule4} (Line 31).
According to Lemma~\ref{le:rule2}, we can further update the candidate set after computing the followers of edges (Line 7).

\section{Upper Bound Based Solution}
\label{sec:upper}

The group based algorithm reduces the size of candidate set  by excluding the edges in the same $k$-support group and the followers of $k$-support groups, which greatly accelerates the baseline method.
However, for each candidate edge, we still need lots of computation to find its followers.
Given an edge, if we can obtain the upper bound of its follower size, then we can speed up the search
by pruning unpromising candidates. In this section, we present a novel method to efficiently
calculate the upper bound required.

\subsection{Upper Bound Derivation}

Before introducing the lemma, we first present some basic definitions. Recall that $\tau(e)$ denotes the trussness of $e$.

\begin{definition}[\textbf{$k$-triangle}]
A triangle $\triangle_{uvw}$ is a k-triangle, if the trussness of each edge is no less than $k$.
\end{definition}

\begin{definition}[\textbf{$k$-triangle connectivity}]
Two triangles $\triangle_{s}$ and $\triangle_{t}$ are $k$-triangle connected, denoted as $\triangle_s \stackrel{k}\leftrightarrow \triangle_t$, if there exists a sequence of $\theta \geq 2$ triangles $\triangle_1, ..., \triangle_{\theta}$ with $\triangle_s = \triangle_1, \triangle_t = \triangle_{\theta}$. For $i \in [1,\theta)$, $\triangle_i \cap \triangle_{i+1} = {e}$ and $\tau(e) = k$.
\end{definition}

We say two edges $e, e'$ are $k$-triangle connected, denoted as $e \stackrel{k} \leftrightarrow e'$, if and only if 1) $e$ and $e'$ belong to the same $k$-triangle, or 2) $e \in \triangle_s, e' \in \triangle_t$, with $\triangle_s \stackrel{k}\leftrightarrow \triangle_t$.

\begin{definition}[\textbf{$k$-truss group}]
Given a graph $G$ and an integer $k \geq 3$, a subgraph $S$ is a $k$-truss group if it satisfies: 1) $\forall e \in S, \tau(e) = k$. 2) $\forall e, e' \in S, e \stackrel{k}\leftrightarrow e'$. 3) $S$ is maximal, i.e., there is no supergraph of $S$ satisfying conditions 1 and 2.
\end{definition}

Based on the definition of $k$-truss group, Lemma~\ref{le:upper_bound} gives an upper bound of $|F(e, \Truss{k})|$.

\begin{lemma}\label{le:upper_bound}
If $e$ is triangle adjacent with $\theta$ k-truss groups $g_1$, $g_2,..., g_{\theta}$, we have $|F(e, \Truss{k})| \leq \sum\limits_{i=1}^{\theta} |\Edge{g_i}|$.
\end{lemma}

\begin{proof}
Suppose $\Sup{e}{\Truss{k}} = w$, we have $w \geq k-2$, so $e$ is contained by $w$ triangles and is triangle adjacent with $2w$ edges.
We divide the edges which are triangle adjacent with $e$ in $\Truss{k}$ into two parts. 1) $\tau(e') > k$.
Since the deletion of $e$ may cause $\tau(e')$ to decrease at most 1~\cite{DBLP:conf/sigmod/HuangCQTY14,DBLP:journals/pvldb/AkbasZ17}, we have $\tau(e') \geq k$ after deleting $e$, which means $e'$ has no contribution to $F(e, \Truss{k})$. 2) $\tau(e') = k$. Suppose $e' \in g_i$.
The deletion of $e$ can cause trussness of each edge in $g_i$ to decrease at most 1.
Then $e'$ can contribute to $|F(e, \Truss{k})|$ with at most $|\Edge{g_i}|$. Thus, $\sum\limits_{i=1}^{\theta} |\Edge{g_i}|$ is an upper bound of $|F(e, \Truss{k})|$.
\end{proof}

\subsection{Upper Bound Based Algorithm}

Based on Lemma~\ref{le:upper_bound}, we can skip the edges whose upper bound of follower size is less than the best edge in the current iteration. However, given the trussness of each edge, it may still be prohibitive to find the $k$-truss group that contains an edge $e$, since in the worst case we need to explore all the triangles in the graph. To compute the upper bound efficiently, we construct an index to maintain the relationships between edges and their $k$-truss groups.

To find the $k$-truss group for a given edge $e$, we extend the GroupExpansion function in Line 20-32 of Algorithm~\ref{alg:improved}. It also follows the BFS search manner. The difference is that
when we explore an adjacent triangle, it must satisfy the $k$-triangle constraint, and we only enqueue an edge, whose trussness satisfies $k$-triangle connectivity constraint. After finishing the BFS search starting from $e$, its involved $k$-truss groups can be found.

After deleting an edge $e$ in the current iteration, the constructed $k$-truss groups may be
changed. Therefore, we need to update the $k$-truss groups for the next iteration. The update algorithm
consists of two parts, \ie update the trussness and update the groups affected by the changed trussness.
To update the edge trussness, we apply the algorithm in~\cite{DBLP:conf/sigmod/HuangCQTY14}, which can
efficiently update the edge trussness after deleting an edge $e$.
Given the edges with changed trussness, we first find the subgraph induced by these edges. Then we reconstruct the $k$-truss groups for the induced subgraph
and update the original ones.
Based on the $k$-truss groups constructed, we can compute the upper bound of followers for edges efficiently. The final algorithm, named \textbf{UP-Edge}, integrates all the techniques proposed in Section~\ref{sec:framework} and~\ref{sec:upper}.

\section{Experiment}
\label{sec:experiment}

\subsection{Experiment Setting}

In the experiments, we implement and evaluate the following algorithms.
1) Exact: naive algorithm that enumerates all the combinations.
2) Support: in each iteration, it selects the edge
that is triangle adjacent with the edge with minimum support in the $k$-truss.
3) Baseline: baseline algorithm in Section~2.2.
4) GP-Edge: group based algorithm in Section~3.
5) UP-Edge: upper bound based algorithm in Section 4.

We employ 9 real social networks (\ie Bitcoin-alpha, Email-Eu-core, Facebook, Brightkite, Gowalla, DBLP, Youtube, Orkut, LiveJournal) to evaluate the performance of the proposed methods. The datasets are public available\footnote{\url{https://snap.stanford.edu/data/}, \url{https://dblp.org/xml/release/}}.
Since the Exact algorithm is too slow, we only run Exact algorithm on Email-Eu-core and Bitcoin-alpha dataset.

Since the properties of datasets are quite different, we set the default $k$ as 10 for 4 datasets (Gowalla, Youtube, Brightkite, DBLP) and set the default $k$ as 20 for 3 datasets (Facebook, LiveJournal, Orkut). We set default $b$ as 5 for all datasets.
All the programs are implemented in C++. All the experiments are performed on a machine with an Intel Xeon 2.20 GHz CPU and 128 GB memory running Linux.

\subsection{Effectiveness Evaluation}
To evaluate the effectiveness of the proposed methods, we report the number of followers by deleting $b$ edges.
Since UP-Edge only accelerates the speed of Baseline and GP-Edge,
we only report the results of UP-Edge here. Due to the huge time cost of Exact, we show the result on 3 datasets, that is, Bitcoin-alpha, Email-Eu-core and Artificial network (generated by GTGraph with 500 nodes and 5000 edges).

We set $k=11$ and $8$ for Bitcoin-alpha and Artificial network respectively, and vary $b$ from 1 to 4.
In Figure~\ref{eff:bitcoin}, we can see that there is only a slight drop when b=3. In Figure~\ref{eff:artificial}, there is only a small drop when b=4. In Figure~\ref{eff:email}, as we can see, UP-Edge also shows comparable
results with Exact and they all outperform Support significantly. Similar results can be observed
in Figure~\ref{eff:all}-\ref{eff:vary_k} over all the datasets and the selected datasets.
Figure~\ref{eff:vary_b} and~\ref{eff:vary_k} show the results on LiveJournal by varying $b$ and $k$. As observed, the number of followers for the two algorithms are positive correlated with $b$, and $k$ has a great impact on follower size.

\begin{figure}[t]
\begin{center}
	\subfigure[\small{Bitcoin-alpha (k=11)}]{
    	\includegraphics[width=0.46\columnwidth]{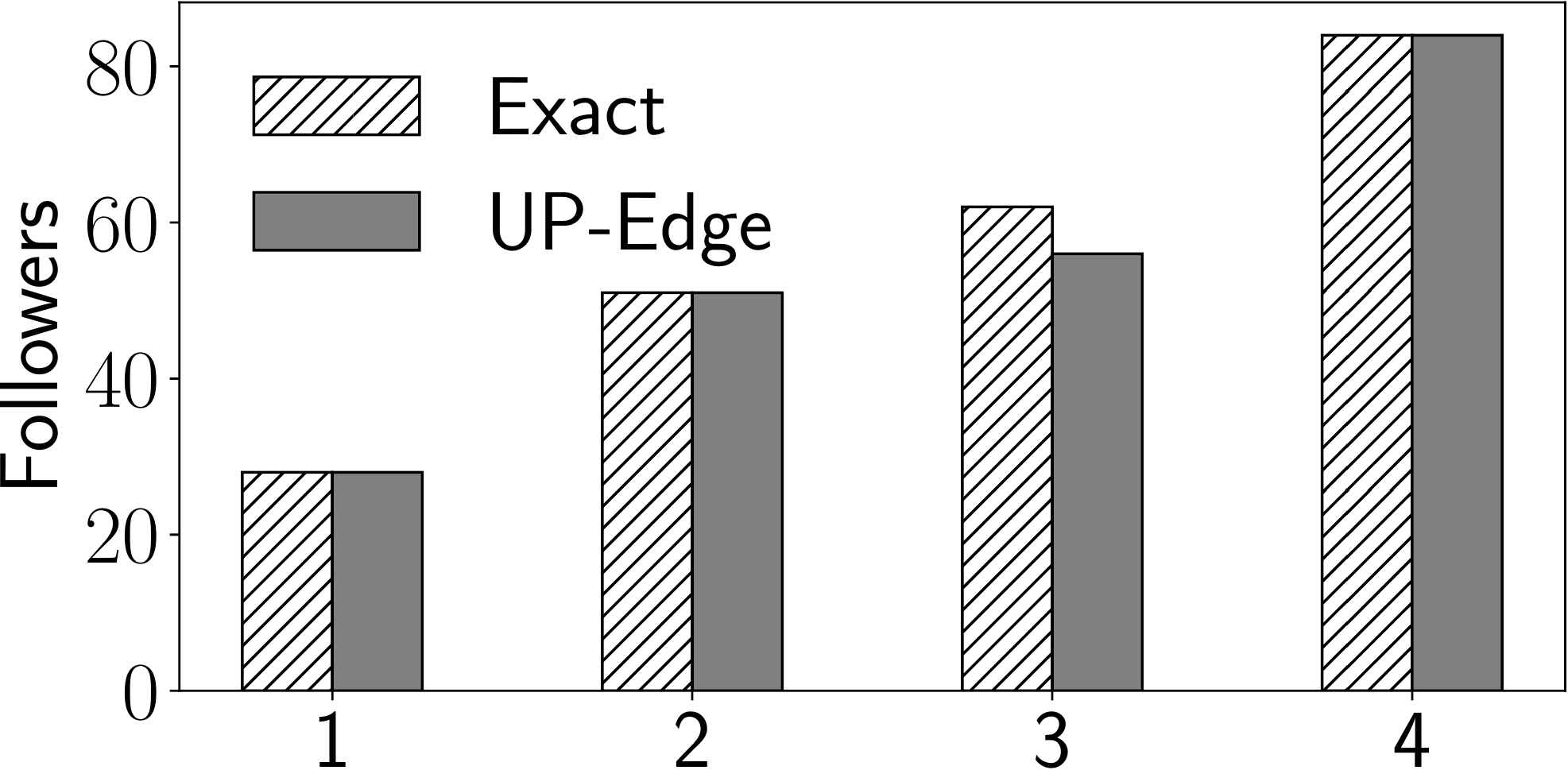}
    	\label{eff:bitcoin}
  	}
  \subfigure[\small{Artificial (k=8)}]{
    \includegraphics[width=0.47\columnwidth]{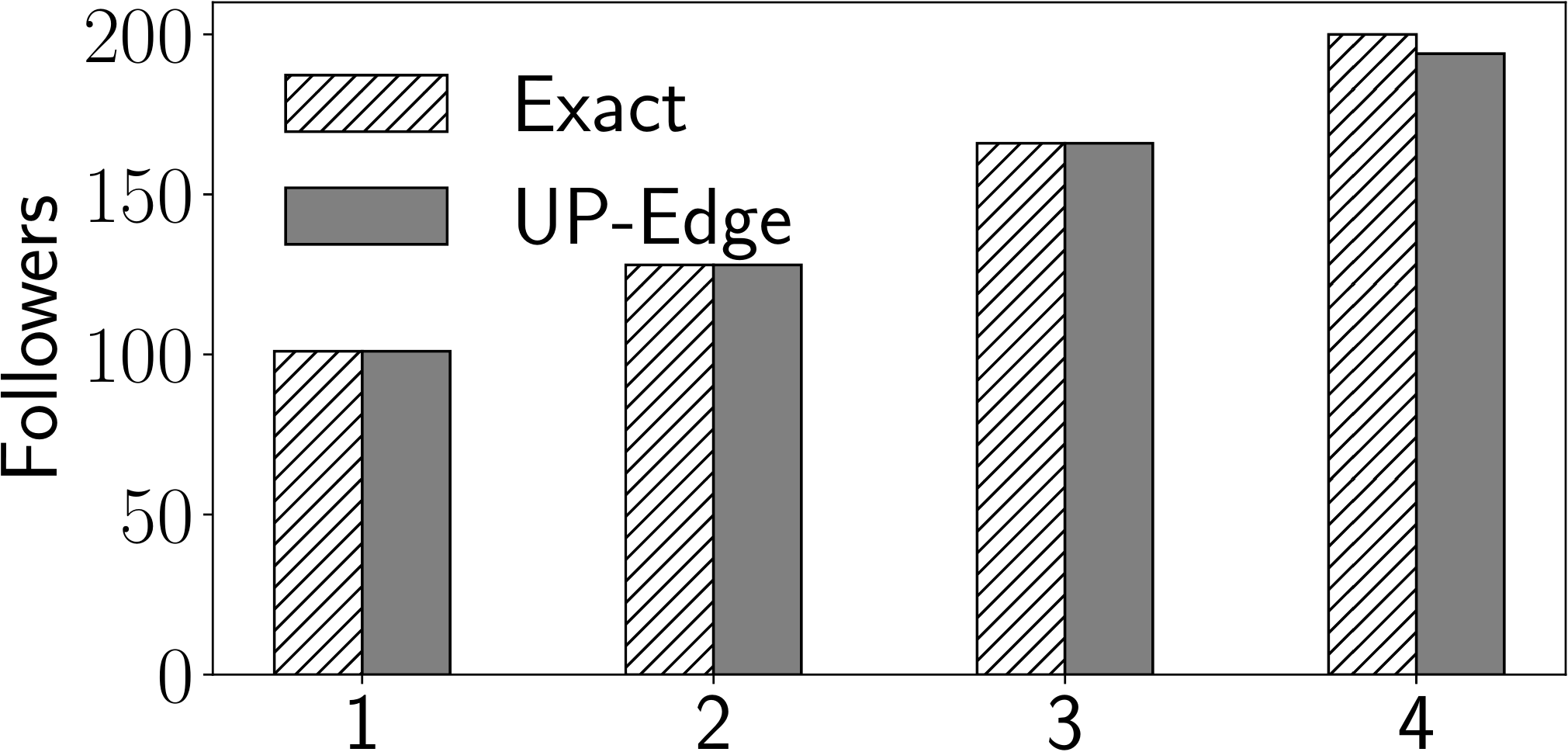}
    \label{eff:artificial}
  }

  \subfigure[\small{Email-Eu-core (b=2)}]{
    \includegraphics[width=0.47\columnwidth]{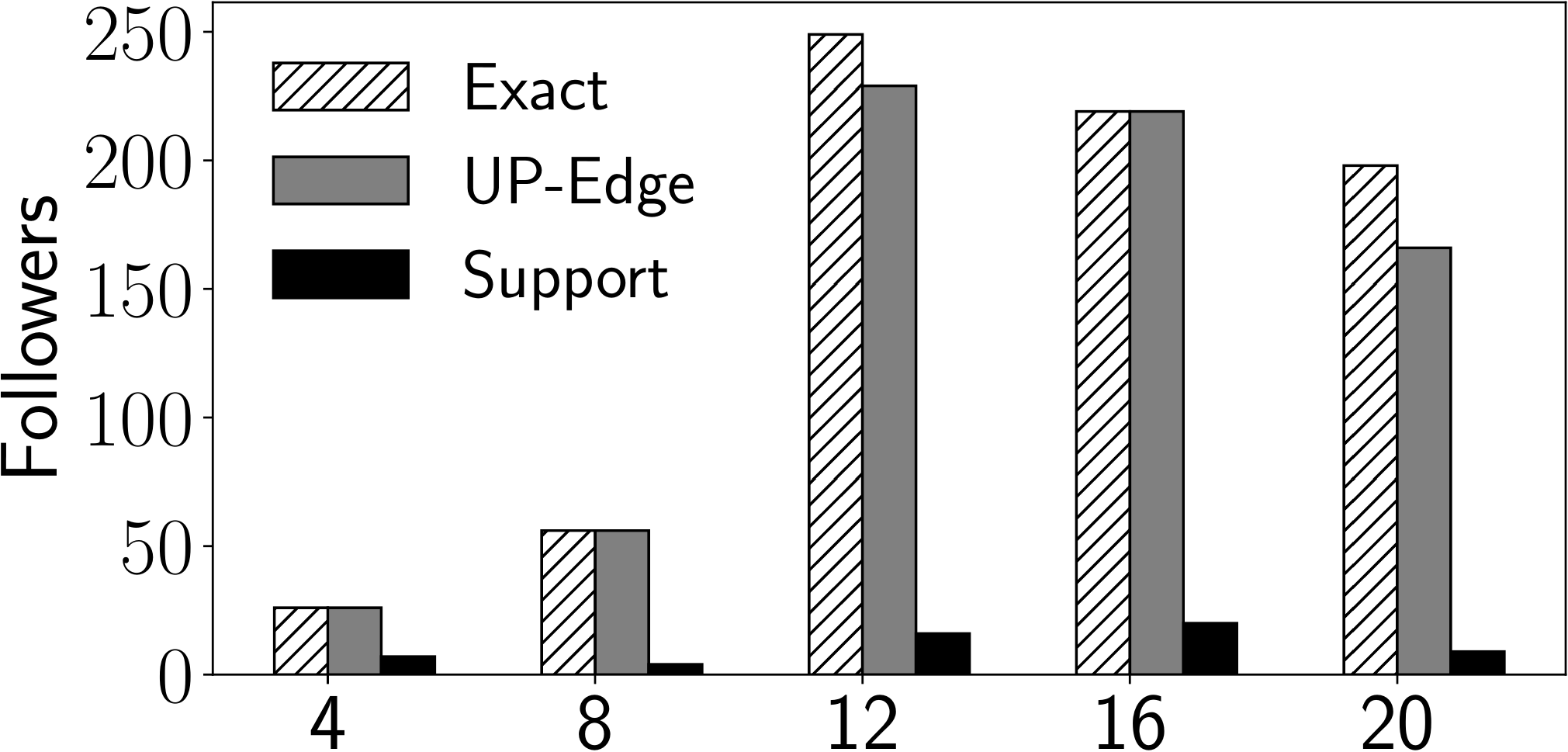}
    \label{eff:email}
  }
  \subfigure[\small{All datasets}]{
    \includegraphics[width=0.47\columnwidth]{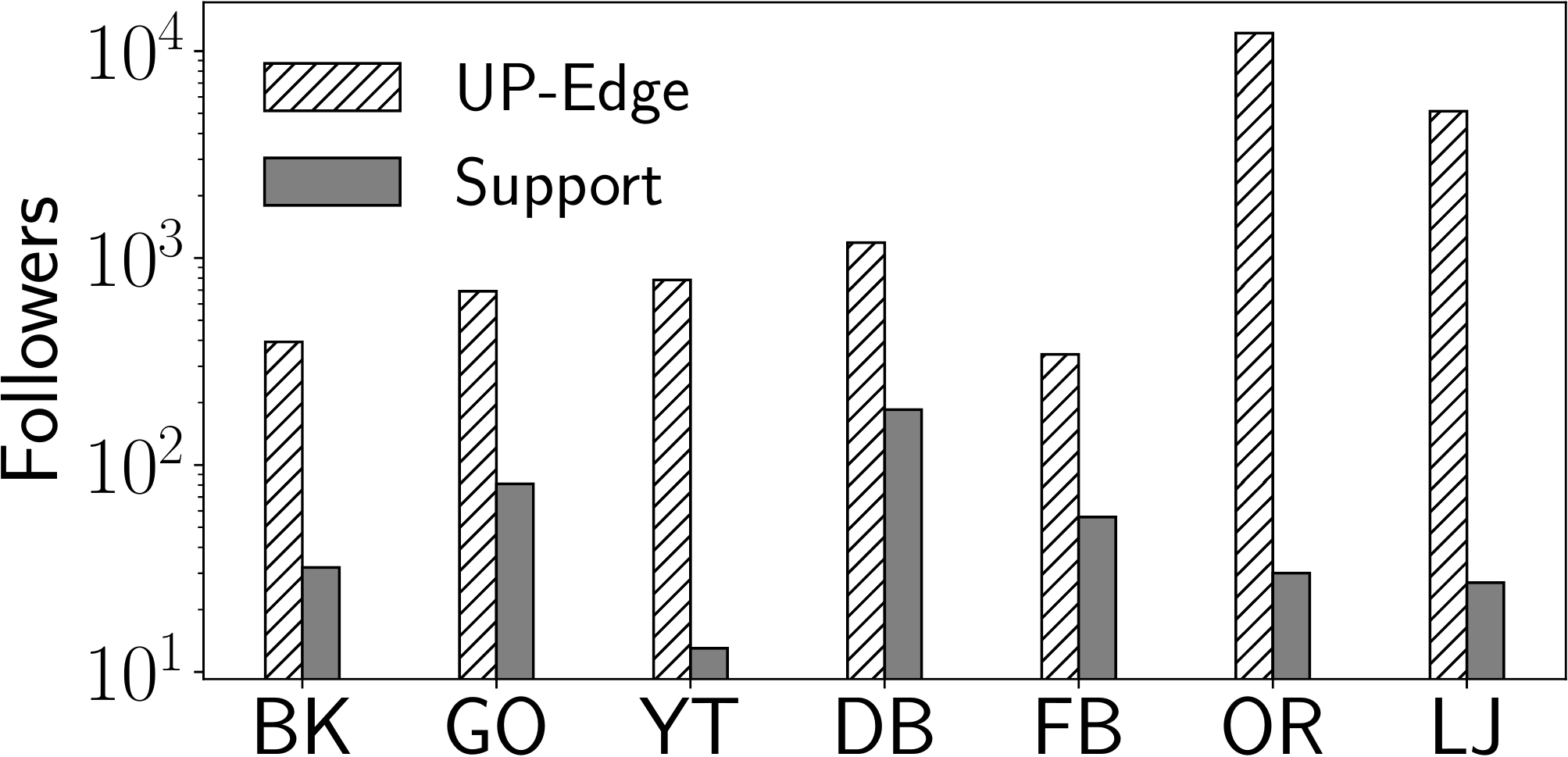}
    \label{eff:all}
  }
	\subfigure[\small{LiveJournal (vary $b$)}]{
    \includegraphics[width=0.47\columnwidth]{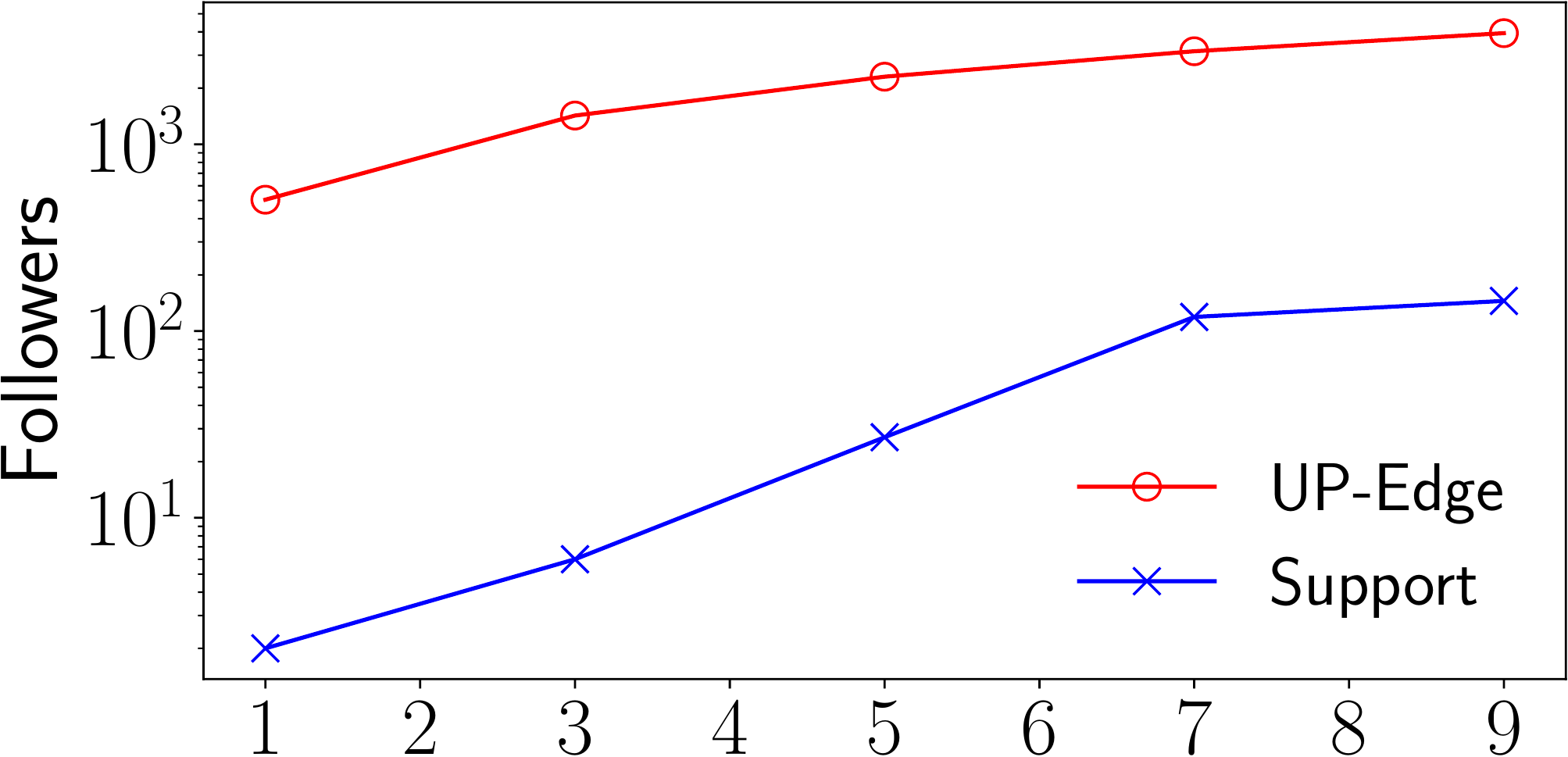}
    \label{eff:vary_b}
  }
  \subfigure[\small{LiveJournal (vary $k$)}]{
    \includegraphics[width=0.47\columnwidth]{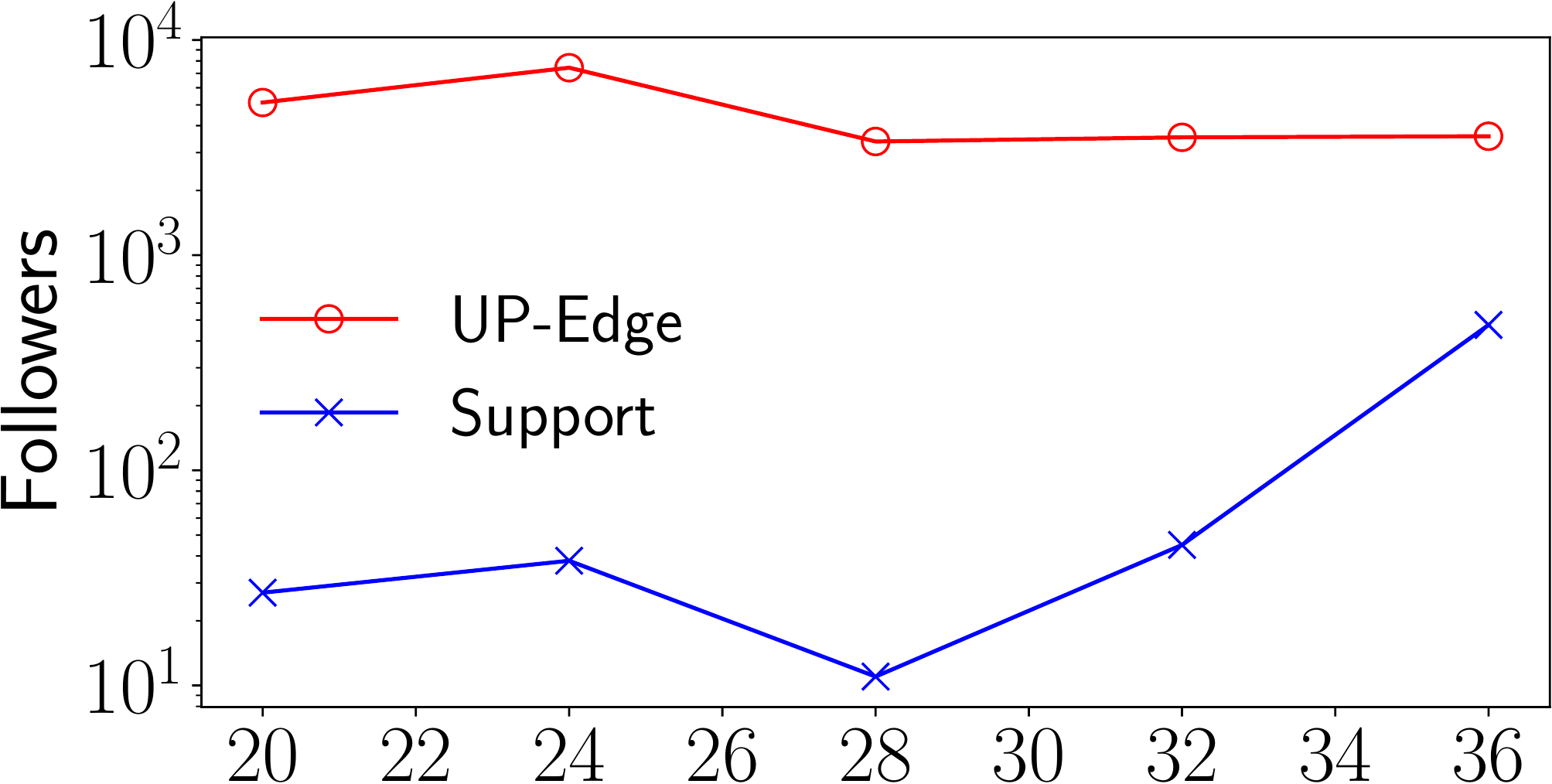}
    \label{eff:vary_k}
  }
 \end{center}
\caption{Effectiveness Evaluation}
\label{fig:followers}
\end{figure}

\begin{figure}[t]
\centering
\includegraphics[width=0.96\linewidth]{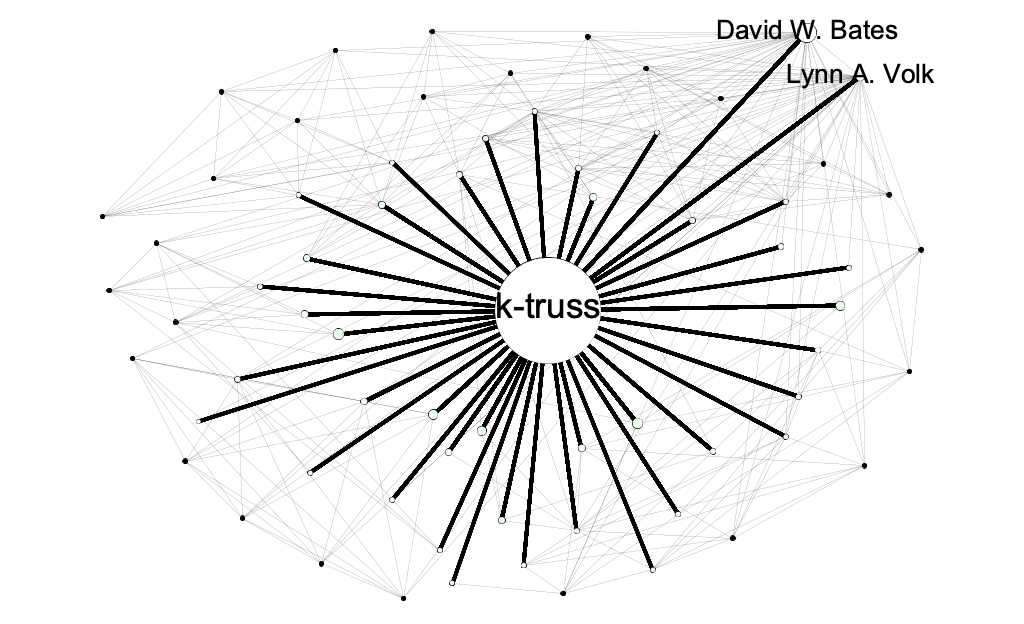}
\caption{Case study on DBLP, k=10, b=1}
\label{fig:case}
\end{figure}

Figure~\ref{fig:case} shows a case study on DBLP with $k=10, b=1$. We can see that the edge between Lynn A. Volk and David W. Bates is the most pivotal relationship. This edge has 264 followers (grey edges in the figure). It is interesting that most followers have no direct connection with them.

\subsection{Efficiency Evaluation}
To evaluate the efficiency, we compare the response time of UP-Edge and GP-Edge with Baseline .
We first conduct the experiments on all the datasets with default settings. Figure~\ref{runtime:alldata} shows the response time of the three algorithms.
We can see that UP-Edge and GP-Edge significantly outperform Baseline in all the datasets because of the pruning techniques developed.
UP-Edge is faster than GP-Edge due to the contribution of upper bound derived.
Figure~\ref{fig:runtime} shows the results conducted on LiveJournal by varying $b$ and $k$.
We can see that when $b$ grows, the response time increases since more edges need to be selected. When $k$ grows, the response time decreases since the searching space becomes smaller.

\begin{figure}[t]
\centering
\includegraphics[width=0.96\linewidth]{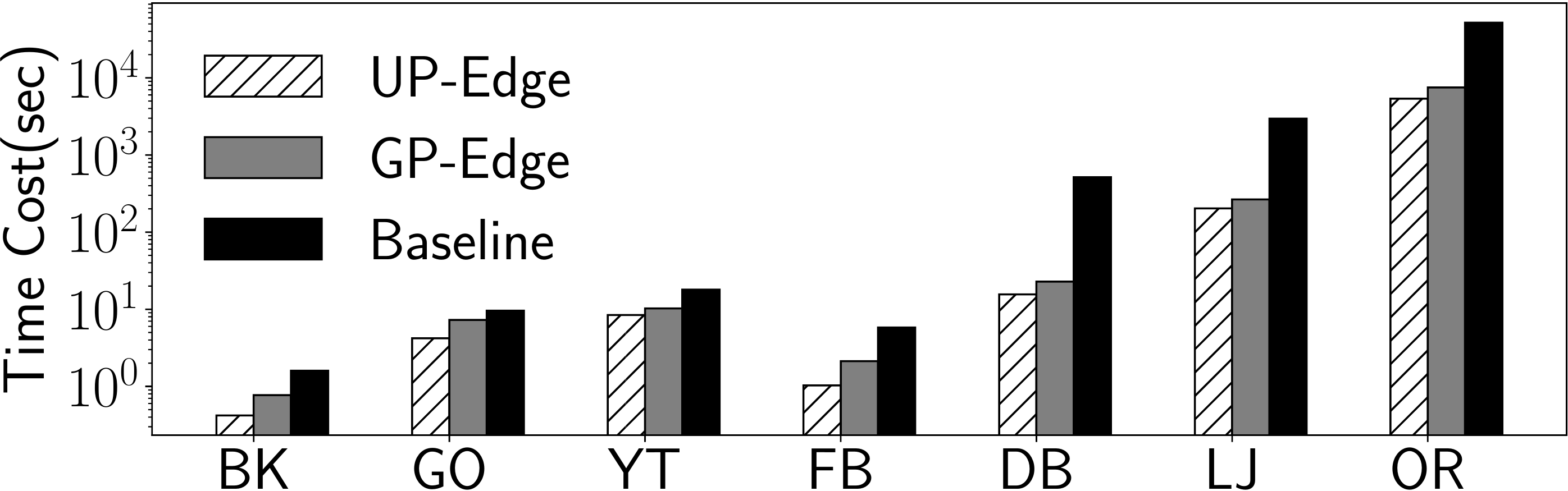}
\caption{Running time on all datasets}
\label{runtime:alldata}
\end{figure}

\begin{figure}[t]
\begin{center}
  \subfigure[\small{LiveJournal (vary $b$)}]{
    \includegraphics[width=0.47\columnwidth]{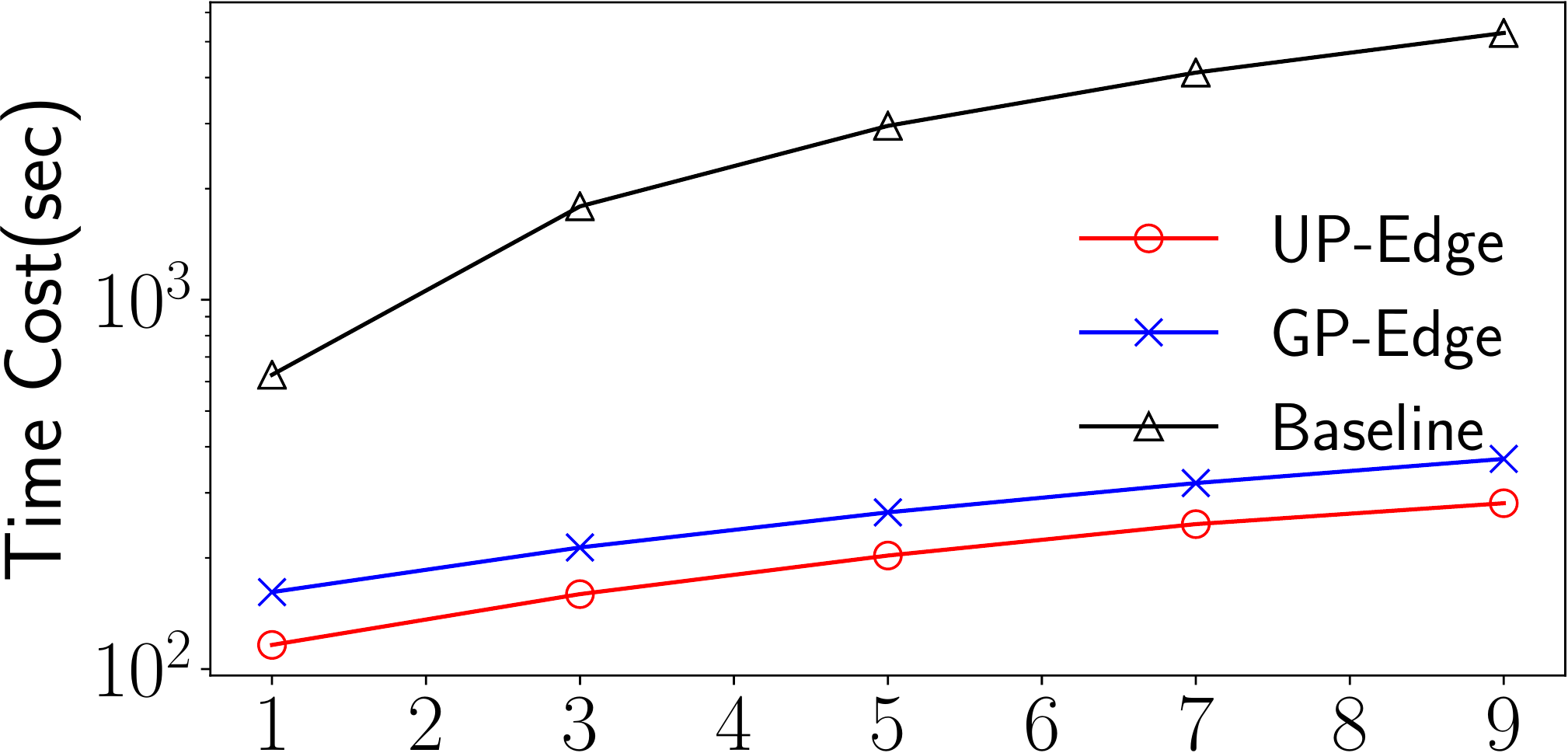}
    \label{runtime:livejournal_vary_b}
  }
  \subfigure[\small{LiveJournal (vary $k$)}]{
    \includegraphics[width=0.47\columnwidth]{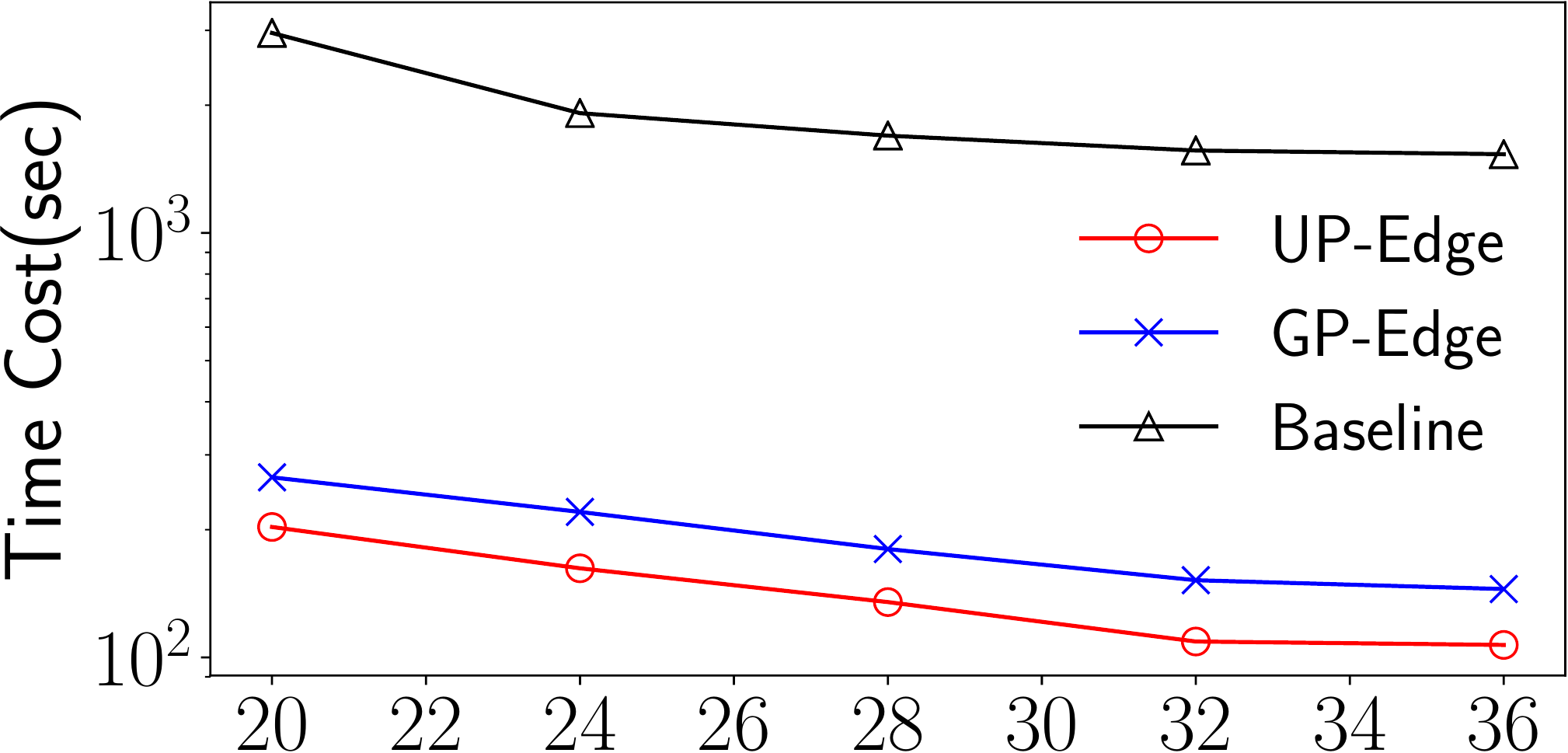}
    \label{runtime:livejournal_vary_k}
  }
\end{center}
\caption{Efficiency Evaluation}
\label{fig:runtime}
\end{figure}

\section{Related Work}
\label{sec:rel}

Graph processing has been a hot topic in many areas recently, which usually requires much more computation comparing with some traditional queries~\cite{luo2008spark,wang2010mapdupreducer,wang2015ap}.
Cohesive subgraph identification is of great importance to social network analysis.
In the literature, different definitions of cohesive subgraphs are proposed, such as $k$-core~\cite{Seidman1983Network,DBLP:conf/icde/WenQZLY16}, $k$-truss~\cite{Huang:2017:ACS:3099622.3099626}, clique~\cite{tsourakakis2013denser}, dense neighborhood graph~\cite{Kreutzer:2018:PKW:3292530.3274652}, etc. In the literature, numerous research is conducted to investigate the $k$-truss decomposition problem under different settings, including in-memory algorithms~\cite{cohen2008trusses}, external-memory algorithms~\cite{Wang:2012:TDM:2311906.2311909}, distributed algorithms~\cite{chen2014distributed}, etc. In some studies, authors leverage the $k$-truss property to mine required
communities~\cite{DBLP:conf/sigmod/HuangCQTY14,Huang:2017:ACS:3099622.3099626}. Huang \etal~\cite{Huang:2016:TDP:2882903.2882913} investigate the truss decomposition problem in uncertain graphs.
Recently, some research focuses on modifying the graph to maximize/minimize the corresponding 
metric~\cite{bhawalkar2015preventing,zhang2017finding,DBLP:conf/cikm/Zhu0WL18,medya2018group}.
Bhawalkar \etal~\cite{bhawalkar2015preventing} propose the anchored $k$-core problem, which tries
to maximize the $k$-core by anchoring $b$ nodes, while Zhang \etal~\cite{zhang2017finding} and Zhu \etal~\cite{DBLP:conf/cikm/Zhu0WL18} investigate the problem of $k$-core minimization by deleting nodes and edges, respectively. 
In~\cite{medya2018group}, Medya \etal try to maximize the node centrality by adding edges to the graph. 
However, these techniques cannot be extended for our problem.
\section{Conclusion}
\label{sec:conc}

In this paper, we study the \emph{k}-truss minimization problem. We first formally define the problem. Due to the hardness of the problem, a greedy baseline  algorithm is proposed. To speed up the search, different pruning techniques are developed. In addition,
an upper bound based strategy is presented by leveraging the $k$-truss group concept. 
Lastly, we conduct extensive experiments on real-world social networks to demonstrate the advantage of the proposed techniques.

\bibliographystyle{named}
\bibliography{ms}
\end{document}